\documentclass[12pt]{article}
\usepackage{amsmath,amsthm,hyperref}
\hypersetup{
    colorlinks,%
    citecolor=blue,%
    filecolor=blue,%
    linkcolor=blue,%
    urlcolor=blue
}
\usepackage{amssymb,verbatim}
\usepackage{fix2col,listings,fancyvrb}
\usepackage{multicol}
\usepackage{graphicx,epsfig}
\usepackage{caption}
\usepackage{stmaryrd}
\usepackage{setspace}
\usepackage{ulem}
\usepackage{newlfont}
\usepackage{epsfig,graphics}
\usepackage{fancybox}
\usepackage{listings}
\usepackage{caption}
\usepackage{graphicx,epsfig}
\usepackage{caption}
\usepackage{makeidx}
\makeindex

\usepackage{algorithm}
\usepackage{algpseudocode}
\usepackage{hyperref}
\hypersetup{
    colorlinks,%
    citecolor=blue,%
    filecolor=blue,%
    linkcolor=blue,%
    urlcolor=blue
}

\newcommand{\dom}{\mbox{dom}}

\newtheorem{theorem}{Theorem}[section]

\newtheorem{proposition}[theorem]{Proposition}

\newtheorem{definition}[theorem]{Definition}
\newtheorem{remark}[theorem]{Remark}

\newtheorem{notation}[theorem]{Notation}

\newcommand{\jump}{\mathsf{jump}}
\newcommand{\inv}{\mathsf{inv}}
\newcommand{\init}{\mathsf{init}}

\newcommand{\unsafe}{\mathsf{unsafe}}

\newcommand{\lrf}{\mathcal{L}_{\mathbb{R}_{\mathcal{F}}}}

\newcommand{\prunefwd}{\mathrm{Prune}_{\mathrm{fwd}}}
\newcommand{\prunebwd}{\mathrm{Prune}_{\mathrm{bwd}}}
\newcommand{\prunetime}{\mathrm{Prune}_{\mathrm{time}}}
\newcommand{\BXt}{\ensuremath{B_{\vec{x_t}}}}
\newcommand{\BXz}{\ensuremath{B_{\vec{x_0}}}}

\DeclareMathOperator{\Hull}{Hull}

\title{Satisfiability Modulo ODEs}
\author{Sicun Gao \and Soonho Kong \and Edmund Clarke}

\begin{document}
\maketitle

\begin{abstract}
We study SMT problems over the reals containing ordinary differential equations. They are important for formal verification of realistic hybrid systems and embedded software. We develop $\delta$-complete algorithms for SMT formulas that are purely existentially quantified, as well as $\exists\forall$-formulas whose universal quantification is restricted to the time variables. We demonstrate scalability of the algorithms, as implemented in our open-source solver {\sf dReal}, on SMT benchmarks with several hundred nonlinear ODEs and variables.
 \end{abstract}



\section{Introduction}

Hybrid systems tightly combine finite automata and continuous dynamics. In most cases, the continuous components are specified by ordinary differential equations (ODEs). Thus, formal verification of general hybrid systems requires reasoning about logic formulas over the reals that contain ODE constraints. This problem is considered very difficult and has not been investigated in the context of decision procedures until recently~\cite{Eggers2008,Eggers2011,DBLP:journals/sttt/IshiiUH11}. It is believed that current techniques are not powerful enough to handle formulas that arise from formal verification of realistic hybrid systems, which typically contain many nonlinear ODEs and other constraints.

Since the first-order theory over the reals with trigonometric functions is already undecidable, solving formulas with general ODEs seems inherently impossible. We have resolved much of this theoretical difficulty by proposing the study of {\em $\delta$-complete} decision procedures for such formulas~\cite{DBLP:conf/lics/GaoAC12}. An algorithm is $\delta$-complete for a set of SMT formulas, where $\delta$ is an arbitrary positive rational number, if it correctly decides whether a formula is unsatisfiable or $\delta$-satisfiable. Here, a formula is $\delta$-satisfiable if, under some {\em $\delta$-perturbations}, a syntactic variant of the original formula is satisfiable~\cite{DBLP:conf/cade/GaoAC12}. We have shown that $\delta$-complete decision procedures are suitable for various formal verification tasks~\cite{DBLP:conf/cade/GaoAC12,DBLP:conf/lics/GaoAC12}. We have also proved that $\delta$-complete decision procedures exist for SMT problems over the reals with Lipschitz-continuous ODEs. Such results serve as a theoretical foundation for developing practical decision procedures for the SMT problem.

In this paper we study practical $\delta$-complete algorithms for SMT formulas over the reals with ODEs. We show that such algorithms can be made powerful enough to scale to realistic benchmark formulas with several hundred nonlinear ODEs.

We develop decision procedures for the problem following a standard DPLL(ICP) framework, which relies on constraint solving algorithms as studied in Interval Constraint Propagation (ICP)~\cite{handbookICP}. In this framework, for any ODE system we can consider its solution function $\vec x_t = \vec f(t, \vec x_0)$ as a constraint between the initial variables $\vec x_0$, time variable $t$, and the final state variables $\vec x_t$. We define pruning operators that take interval assignments on $\vec x_0$, $t$, and $\vec x_t$ as inputs, and output refined interval assignments on these variables. We formally prove that the proposed algorithms are $\delta$-complete. Beyond standard SMT problems where all variables are existentially quantified, we also study $\exists\forall$-formulas under the restriction that the universal quantifications are limited to the time variables (we call them $\exists\forall^t$-formulas). Such formulas have been an obstacle in SMT-based verification of hybrid systems~\cite{DBLP:conf/fmcad/CimattiMT12,DBLP:conf/aaai/CimattiMT12}. 

In brief, this paper makes the following contributions:

\begin{itemize}
\item We formalize the SMT problem over the reals with general Lipschitz-contiunous ODEs, and illustrate its expressiveness by encoding various standard problems concerning ODEs: initial and boundary value problems, parameter synthesis problems, differential algebraic equations, and bounded model checking of hybrid systems. In some cases, $\exists\forall^t$-formulas are needed.
\item We propose algorithms for solving SMT with ODEs, using ODE constraints to design pruning operators in a branch-and-prune framework. We handle both standard SMT problems with only existentially quantified variables, as well as $\exists\forall^t$-formulas. We prove that the algorithms are $\delta$-complete.
\item We demonstrate the scalability of the algorithms, as implemented in our open-source solver {\sf dReal}~\cite{DBLP:conf/cade/GaoKC13}, on realistic benchmarks encoding formal verification problems for several nonlinear hybrid systems.
\end{itemize}

The paper is organized as follows. In Section~\ref{language}, we define the SMT problem with ODEs and show how it can encode various standard problems with ODEs. In Section~\ref{algorithms}, we propose algorithms in the DPLL(ICP) framework for solving fully existentially quantified formulas as well as $\exists\forall^t$ formulas. In Section~\ref{experiments} we show experimental results.

\paragraph{Related Work.} 
Solving real constraints with ODEs has a wide range of applications, and much previous work exists for classes with special structures in different paradigms~\cite{DBLP:conf/cp/CruzB03,Granvilliers:2005:PEU:1039891.1039931,Lin07guaranteedstate}. Recently~\cite{DBLP:conf/cp/GoldsztejnMEH10} proposed a more general constraint solving framework, focusing on the formulation of the problem in the standard CP framework. On the SMT solving side, several authors have considered logical combinations of ODE constraints and proposed partial decision procedures~\cite{Eggers2008,Eggers2011,DBLP:journals/sttt/IshiiUH11}. We aim to extend and formalize existing algorithms for a general SMT theory with ODES, and formally prove that they can be made $\delta$-complete. In terms of practical performance, the proposed algorithms are made scalable to various benchmarks that contain hundreds of nonlinear ODEs and variables. 

\section{SMT over the Reals with ODEs}~\label{language}
\vspace{-0.8cm}
\subsection{Computable Real Functions}
As studied in Computable Analysis~\cite{CAbook,Kobook}, we can encode real numbers as infinite strings, and develop a computability theory of real functions using Turing machines that perform operations using oracles encoding real numbers.  We briefly review definitions and results of importance to us. Throughout the paper we use $||\cdot||$ to denote the max norm $||\cdot||_{\infty}$ over $\mathbb{R}^n$ for various $n$. First, a {\em name} of a real number is a sequence of rational numbers converging to it:
\begin{definition}[Names]
A name of $a\in \mathbb{R}$ is any function $\mathcal{\gamma}_a: \mathbb{N}\rightarrow \mathbb{Q}$ that satisfies:  for any $i\in \mathbb{N}$, $|\gamma_a(i) - a|<2^{-i}$. For $\vec a\in \mathbb{R}^n$, $\gamma_{\vec a}(i) = \langle \gamma_{a_1}(i), ..., \gamma_{a_n}(i)\rangle$. We write the set of all possible names for $\vec a$ as $\Gamma(\vec a)$.
\end{definition}

Next, a real function $f$ is {\em computable} if there is a Turing machine that can use any argument $x$ of $f$ as an oracle, and compute the value of $f(x)$ up to an arbitrary precision $2^{-i}$, where $i\in\mathbb{N}$. Formally: 
\begin{definition}[Computable Functions]\label{comp_functions}
We say $f:\subseteq\mathbb{R}^n\rightarrow \mathbb{R}$ is computable if there exists an oracle Turing machine $\mathcal{M}_f$ such that for any $\vec x\in \dom(f)$, any name $\gamma_{\vec x}$ of $\vec x$, and any $i\in \mathbb{N}$, the machine uses $\gamma_{\vec x}$ as an oracle and $i$ as an input to compute a rational number $M_f^{\gamma_{\vec x}}(i)$ satisfying $|M_f^{\gamma_{\vec x}}(i) - f(\vec x)|<2^{-i}.$
\end{definition}

The definition requires that for any $\vec x\in \dom(f)$, with access to an arbitrary oracle encoding the name $\gamma_{\vec x}$ of $\vec x$, $M_f$ outputs a $2^{-i}$-approximation of $f(\vec x)$. In other words, the sequence $M_f^{\gamma_{\vec x}}(1), M_f^{\gamma_{\vec x}}(2), ...$ is a name of $f(\vec x)$. Intuitively, $f$ is computable if an arbitrarily good approximation of $f(\vec x)$ can be obtained using any good enough approximation to any $\vec x\in\dom(f)$.
Most common continuous real functions are computable~\cite{CAbook}. Addition, multiplication, absolute value, $\min$, $\max$, $\exp$, $\sin$. Compositions of computable functions are computable. In particular, solution functions of Lipschitz-continuous ordinary differential equations are computable, as we explain next. 

\subsection{Solution Functions of ODEs}

We now show that the framework of computable functions allows us to consider solution functions of ODE systems.
\begin{notation}
We use $\vec x= \vec y$ between $n$-dimensional vectors to denote the system of equations $x_i=y_i$ for $1\leq i\leq n$.
\end{notation}
Let $D\subseteq \mathbb{R}^n$ be compact and $g_i: D\rightarrow \mathbb{R}$ be $n$ Lipschitz-continuous functions, which means that for some constant $c_i\in \mathbb{R}^+$ ($1 \leq i \leq n$), for all
 $\vec x_1, \vec x_2\in D$,
$$|g_i(\vec x_1)-g_i(\vec x_2)|\leq c_i||\vec x_1-\vec x_2||.$$
Let $t$ be a variable over $\mathbb{R}$. We consider the first-order autonomous ODE system
\begin{eqnarray}\label{ivpfirst}
\frac{d \vec y}{dt} = \vec g(\vec y(t, \vec x_0)) \mbox{ and }\vec y(0, \vec x_0) = \vec x_0
\end{eqnarray}
where $\vec x_0\in D$. Here, each
\begin{eqnarray}\label{ivp_solution_first}
y_i: \mathbb{R}\times D\rightarrow \mathbb{R}
\end{eqnarray}
is called the $i$-th solution function of the ODE system (\ref{ivpfirst}). A key result in computable analysis is that these solution functions are computable, in the sense of Definition~\ref{comp_functions}: 
\begin{proposition}[\cite{Kobook}]
The solution functions $\vec y$ in the form of (\ref{ivp_solution_first}) of the ODE system (\ref{ivpfirst}) are computable over $\mathbb{R}\times D$.
\end{proposition}
To see why this is true, recall that for any $t\in\mathbb{R}$ and $\vec x_0\in D$, the value of the solution function follows the Picard-Lindel\"of form:
$$\vec y(t,\vec x_0) = \int_{0}^t \vec g(\vec y(s,\vec x_0))ds + \vec x_0.$$
Approximations of the right-hand side of the equation can be computed by finite sums, theoretically up to an arbitrary precision.

\subsection{SMT Problems and $\delta$-Complete Decision Procedures}

We now let $\mathcal{F}$ denote an arbitrary collection of computable real functions, which can naturally contain solution functions of ODE systems in the form of (\ref{ivp_solution_first}). Let $\mathcal{L}_{\mathcal{F}}$ denote the first-order signature $\langle \mathcal{F}, <\rangle$, where constants are seen as 0-ary functions in $\mathcal{F}$. Let $\mathbb{R}_{\mathcal{F}}$ be the  structure $\langle \mathbb{R}, \mathcal{F}^{\mathbb{R}}, <^{\mathbb{R}}\rangle$ that interprets $\lrf$-formulas in the standard way. We focus on formulas whose variables take values from bounded domains, which can be defined using bounded quantifiers:
\begin{definition}[Bounded Quantifiers]
The bounded quantifiers $\exists^{[u,v]}$ and $\forall^{[u,v]}$ are defined as
\begin{align*}
\exists^{[u,v]}x.\varphi &=_{df}\exists x. ( u \leq x \land x \leq v \wedge
\varphi),\\
\forall^{[u,v]}x.\varphi &=_{df} \forall x. ( (u \leq x \land x \leq v)
\rightarrow \varphi),
\end{align*}
where $u$ and $v$ denote $\lrf$ terms, whose variables only
contain free variables in $\varphi$ excluding $x$. It is easy to check that
$\exists^{[u,v]}x. \varphi \leftrightarrow \neg \forall^{[u,v]}x. \neg\varphi$.
\end{definition}
The key definition in our framework is {\em $\delta$-variants} of first-order formulas: 
\begin{definition}[$\delta$-Variants]\label{variants}
Let $\delta\in \mathbb{Q}^+\cup\{0\}$, and $\varphi$ a bounded
$\lrf$-sentence of the standard form
$$\varphi: \ Q_1^{I_1}x_1\cdots Q_n^{I_n}x_n\;\psi[t_i(\vec x)>0;
t_j(\vec x)\geq 0],$$ where $i\in\{1,...k\}$ and $j\in\{k+1,...,m\}$.  Note that negations are represented by sign changes on the terms. The {\em
$\delta$-weakening} $\varphi^{\delta}$ of $\varphi$ is
defined as the result of replacing each atom $t_i > 0$ by $t_i >
-\delta$ and $t_j \geq 0$ by $t_j \geq -\delta$. That is,
$$\varphi^{-\delta}:\ Q_1^{I_1}x_1\cdots Q_n^{I_n}x_n\;\psi[t_i(\vec x)>-\delta; t_j(\vec x)\geq -\delta].$$
\end{definition}
The SMT problem is standardly defined as deciding satisfiability of quantifier-free formulas, which is equivalent to deciding the truth value of fully existentially quantified sentences. We will also consider formulas that are partially universally quantified. Thus, we consider both $\Sigma_1$ and $\Sigma_2$ formulas here.
\begin{definition}[Bounded $\Sigma_1$- and $\Sigma_2$-SMT Problems]
A $\Sigma_1$-SMT problem is a formula of the form
$$\exists^{I_1}x_1\cdots\exists^{I_n}x_n.\varphi(\vec x)$$
and a $\Sigma_2$-SMT problem is of the form
$$\exists^{I_1}x_1\cdots\exists^{I_n}x_n\forall^{I_{n+1}}x_{n+1}\cdots\forall^{I_{m}}x_m.\varphi(\vec x).$$
In both cases $\varphi(\vec x)$ is a quantifier-free $\lrf$-formula.
\end{definition}
\begin{definition}[$\delta$-Completeness~\cite{DBLP:conf/cade/GaoAC12}]
Let $S$ be a set of $\lrf$ formulas, and $\delta\in \mathbb{Q}^+$. We say a decision procedure $A$ is $\delta$-complete for $S$, if for any $\varphi\in S$, $A$ correctly returns one of the following answers

$\bullet$ $\varphi$ is false;

$\bullet$ $\varphi^{-\delta}$ is true.

If the two cases overlap, either one is correct.
\end{definition}
We have proved in~\cite{DBLP:conf/lics/GaoAC12} that $\delta$-complete decision procedures exists for arbitrary bounded $\lrf$-sentences. In particular, there exists $\delta$-complete decision procedures for the bounded $\Sigma_1$ and $\Sigma_2$ SMT problems. This serves as the theoretical foundation as well as a correctness requirement for the practical algorithms that we will develop in the following sections.

\subsection{SMT Encoding of Standard Problems with ODEs}
\label{encoding}
In this section, we list several standard problems related to ODE systems and show that they can be easily encoded and generalized through SMT formulas. They motivate the development of decision procedures for the theory. 
\begin{remark}
In all the following cases, solutions to the standard problems are obtained from witnesses for the existentially quantified variables in the SMT formulas.
\end{remark}
\begin{remark}
In the definitions below, when the solution functions $\vec y$ of ODE systems are written as part of a formula, no analytic forms are needed. They are functions included in the signature $\lrf$. 
\end{remark}
\noindent {\bf Generalized Initial Value Problems.} Given an ODE system, the standard initial value problem asks for a solution of the variables at certain time, given a complete assignment to the initial conditions of the system. In the form of SMT formulas, we easily allow the initial conditions to be constrained by arbitrary quantifier-free $\lrf$-formulas:
\begin{definition}[Generalized IVP]
Let $X\subseteq \mathbb{R}^n$ be a compact domain, $T\in \mathbb{R}^+$, and $\vec y: [0,T]\times X\rightarrow X$ be the computable solution functions of an ODE system.  Let $t\in [0,T]$ be an arbitrary constant that represents a time point of interest. The generalized IVP problem is defined by formulas of the form:
$$\exists^{X} x_0\exists^{X} \vec x.\; \varphi(\vec x_0)\wedge \vec x = \vec y(t,\vec x_0),$$
where $\varphi$ is a quantifier-free $\lrf$-formula constraining the initial states $\vec x_0$, and $\vec x$ is the needed value for time point $t$. 
\end{definition}
\noindent{\bf Generalized Boundary Value Problems.} Given an ODE system, the standard boundary value problem is concerned with computing the computable solution function when part of the variables are assigned values at the beginning of flow, and part of the variables as assigned values at the end of the flow. A generalized version as encoded by SMT formulas is:
\begin{definition}[Generalized BVP]
Let $X\subseteq \mathbb{R}^n$ be a compact domain, $T\in \mathbb{R}^+$, and $\vec y: [0,T]\times X\rightarrow X$ be the solution functions of an ODE system. Let $t, t'\in [0,T]$ be two time points of interest. The generalized BVP problem is:
$$
\exists^{X} x_0\exists^{X}\vec x_t\exists^X \vec x.\varphi(\vec x_{0},\vec x_t,t)\wedge \vec x_t = \vec y(t,\vec x_0)\wedge \vec x = \vec y(t',\vec x_0)
$$
where $\varphi$ is a quantifier-free $\lrf$-formula that specifies the boundary conditions. Note that $\vec x$ is the value that we are interested in solving in the chosen time point $t'$.
\end{definition}

\noindent{\bf Data-Fitting and Parameter Synthesis.}
The data fitting problem is the following. Suppose an ODE system has part of its parameters unspecified. Given a sequence of data $(t_1, \vec a_1), ..., (t_k, \vec a_k)$, we need to find the values of the missing parameters of the original ODE system. More formally:
\begin{definition}[Data-Fitting Problems]
Let $X\subseteq \mathbb{R}^n$ and $P\subseteq \mathbb{R}^m$ be compact domains, $T\in \mathbb{R}^+$, and $\vec y(\vec p): [0,T]\times X\rightarrow X$ be the solution functions of an ODE system, where $\vec p\in P$ be a vector of parameters. Let $(t_1,\vec a_1), ..., (t_k, \vec a_k)$ be a sequence of pairs in $[0,T]\times X$. The data-fitting problem is defined by:
$$\exists^P \vec p\exists^X x_0.\; \varphi(\vec x_0)\wedge \vec a_1 = \vec y(\vec p, t_1, \vec x_0)\wedge \cdots \wedge \vec a_k = \vec y(\vec p, t_k, \vec x_0),$$
where a quantifier-free $\varphi$ constraints the initial states $\vec x_0$. 
\end{definition}

\noindent{\bf Differential Algebraic Equations.} DAE problems combine ODEs and algebraic constraints:
\begin{eqnarray}
\frac{d \vec y}{dt} &=& \vec g(\vec y(t, \vec y_0), \vec z)\label{ode_dae}\\
0 &=& \vec h(\vec y, \vec z, t)\label{dae_eqn}
\end{eqnarray}
where $\vec y, \vec y_0\in \mathbb{R}^n$, $\vec z\in \mathbb{R}^m$. To express the problem in $\lrf$, we need to use extra universal quantification to ensure that the algebraic relations hold throughout the time duration. Again, we can also generalize the equation in (\ref{dae_eqn}) to an arbitrary quantifier-free $\lrf$-formula. The problem is encoded as:
\begin{definition}[DAE Problems]
Let $X\subseteq \mathbb{R}^n$ be a compact domain, $T\in \mathbb{R}^+$, and $\vec y: [0,T]\times X\times X\rightarrow X$ be the computable solution functions of the ODE system in (\ref{ode_dae}) parameterized by $\vec z$. Let $h$ be defined by (\ref{dae_eqn}). Let $t\in [0,T]$ be a time point of interest. A DAE problem is defined by the following formula:
$$
\exists^X \vec x_0 \exists^X \vec x \exists^Z \vec z \forall^{[0,t]} t'.\;
\varphi(\vec x_0)\wedge \vec x = \vec y (t, \vec x_0, \vec z) \wedge h(\vec y(\vec x_0, t'), \vec z, t')=0
$$
where a quantifier-free $\varphi$ specifies the initial conditions for $\vec y$, and $\vec x$ is the needed value at time point $t$. 
\end{definition}

\noindent{\bf Bounded Model Checking of Hybrid Systems.} Bounded model checking problems for hybrid systems can be naturally encoded as SMT formulas with ODEs~\cite{Eggers2008,Eggers2011,DBLP:journals/sttt/IshiiUH11,DBLP:conf/fmcad/CimattiMT12,DBLP:conf/aaai/CimattiMT12}. We consider a simple hybrid system to show an example. Let $H$ be an $n$-dimensional 2-mode hybrid system. In mode 1, the flow of the system follows an ODE system whose solution function is $\vec y_1(t, \vec x_0)$, and in mode 2, it follows another solution function $\vec y_2(t, \vec x_0)$. The jump condition from mode 1 to mode 2 is specified by $\jump(\vec x, \vec x')$. The invariants are specified by $\inv_i(\vec x)$ and for mode $i$. Let $\unsafe(\vec x)$ denote an unsafe region. Let the continuous variables be bounded in $X$ and time be bounded in $[0,T]$. Now, if $H$ starts from mode 1 with initial states satisfying $\init(\vec x)$, it can reach the unsafe region after one discrete jump from mode 1 to mode 2, iff the following formula is true:
{\small
\begin{multline*}
\exists^X \vec x_1 \exists^X \vec x_1^t\ \exists^X \vec x_2\exists^X \vec x_2^t\ \exists^{[0,T]}t_1\exists^{[0,T]}t_2\ \forall^{[0,t_1]}t_1'\forall^{[0,t_2]}t_2'.\\
\init(\vec x_1)\wedge\vec x_1^t = \vec y_1(t_1,\vec x_1)\wedge \inv_1(\vec y_1(t_1', \vec x_1))\wedge \jump(\vec x_1^t, \vec x_2)\\
\wedge \vec x_2^t = \vec y_2(t_2, \vec x_2)\wedge \inv_2(\vec y_2(t_2', \vec x_2))\wedge \unsafe(\vec x_2^t).
\end{multline*}
}The encoding can be explained as follows. For each mode, we use two variable vectors $\vec x_i$ and $\vec x_i^t$ to represent the continous flows. $\vec x_i$ denote the starting values of a flow, and $\vec x_i^t$ denotes the final values. In mode $1$, the flow starts with some values in the initial states, specified by $\init(\vec x_1)$. Then, we follow the continuous dynamics in mode $1$, so that $\vec x_1^t$ denotes the final value $\vec x_1^t = \vec y(t_1, \vec x_1)$. Then the system follows the jumping condition and resets the variables from $\vec x_1^t$ to $\vec x_2$ as specified by $\jump(\vec x_1^t, \vec x_2)$. After that, the system follows the flow in mode 2. In the end, we check if the final state $\vec x_2^t$ in mode 2 satisfies the unsafe predicate, $\unsafe(\vec x_2)$. 

\section{Algorithms}\label{algorithms}

\subsection{The ICP framework}

The method of Interval Constraint Propagation (ICP)~\cite{handbookICP} finds solutions of real constraints using a ``branch-and-prune'' method that performs constraint propagation of interval assignments on real variables. The intervals are represented by floating-point end-points. Only over-approximations of the function values are used, which are defined by interval extensions of real functions.
\begin{definition}[Floating-Point Intervals and Hulls]
Let $\mathbb{F}$ denote the finite set of all floating point numbers with symbols $-\infty$ and $+\infty$ under the conventional order $<$. Let
$$\mathbb{IF} = \{[a,b]\subseteq \mathbb{R}: a,b\in \mathbb{F}, a\leq b\}\mbox{ and } \mathbb{BF} = \bigcup_{n=1}^{\infty}\mathbb{IF}^n$$ denote the set of closed real intervals with floating-point endpoints, and the set of {\em boxes} with these intervals, respectively. When $S\subseteq \mathbb{R}^n$ is a set of real numbers, the hull of $S$ is:
$$\Hull(S) = \bigcap \{B\in \mathbb{BF}: S\subseteq B\}.$$
\end{definition}
\begin{definition}[Interval Extension~\cite{handbookICP}]
Suppose $f:\subseteq\mathbb{R}^n\rightarrow \mathbb{R}$ is a real function. An interval extension operator $\sharp(\cdot)$ maps $f$ to a function $\sharp f:\subseteq \mathbb{BF}\rightarrow \mathbb{IF}$, such that
for any $B\in \dom(\sharp f)$, it is always true that $\{f(\vec x):\vec x\in B\}\subseteq \sharp f(B).$
\end{definition}
\begin{algorithm}\label{algo1}
\caption{ICP($f_1,...,f_m, B_0 = I_1^0\times\cdots\times I_n^0, \delta$)}\label{icpalgo}
\begin{algorithmic}[1]
\Statex
    \State $S \gets B_0$
    \While{$S \neq \emptyset$}
        \State $B \gets S.\mathrm{pop}()$
        \While{$\exists 1 \leq i \leq m, B \neq_{\delta} \mathrm{Prune}(B,f_i)$}
        \State $B \gets \mathrm{Prune}(B, f_i)$
        \EndWhile
        \If{$B\neq \emptyset$}
            \If{$\exists 1\leq i\leq n, |\sharp f_i(B)|\geq \delta$}
                \State $\{B_1,B_2\} \gets \mathrm{Branch}(B, i)$
                \State $S.\mathrm{push}(\{B_1,B_2\})$
            \Else
                \State \Return {\sf sat}
            \EndIf
        \EndIf
    \EndWhile
    \State \Return {\sf unsat}
\end{algorithmic}
\end{algorithm}
ICP uses interval extensions of functions to ``prune'' out sets of points that are not in the solution set, and ``branch'' on intervals when such pruning can not be done, until a small enough box that may contain a solution is found. A high-level description of the decision version of ICP is given in Algorithm~\ref{icpalgo}. In Algorithm~\ref{icpalgo}, Branch$(B,i)$ is an operator that returns two smaller boxes $B' = I_1\times\cdots\times I_i'\times\cdots\times I_n$ and $B''=I_1\times \cdots\times I_i''\times \cdots\times I_n$, where $I_i\subseteq I_i'\cup I_i''$. The key component of the algorithm is the $\mathrm{Prune}(B, f)$ operation.  Any operation that contracts the intervals on variables can be seen as pruning, but for correctness we need formal requirements on the pruning operator in ICP.  Basically, we need to require that the interval extensions of the functions converge to the true values of the functions, and that the pruning operations are well-defined, as specified below.  
\begin{definition}[$\delta$-Regular Interval Extensions]
We say an interval extension $\sharp f$ of $f:\mathbb{R}^n\rightarrow \mathbb{R}$ is $\delta$-regular, if for some constant $c\in\mathbb{R}$, for any $B\in \mathbb{R}^n$, $|\sharp f(B)|\leq \max(c||B||, \delta)$. 
\end{definition}
\begin{definition}[Well-defined Pruning Operators~\cite{DBLP:conf/cade/GaoAC12}]\label{well}
Let $\mathcal{F}$ be a collection of real functions, and $\sharp$ be a $\delta$-regular interval extension operator on $\mathcal{F}$. A {\em well-defined (equality) pruning operator} with respect to $\sharp$ is a partial function $\mathrm{Prune}_{\sharp} : \subseteq \mathbb{BF}\times \mathcal{F}\rightarrow \mathbb{BF}$, such that for any $f\in \mathcal{F}$, $B,B'\in \mathbb{BF}$,
\begin{enumerate}
\item $\mathrm{Prune}_{\sharp}(B, f)\subseteq B$;
\item If $\mathrm{Prune}_{\sharp}(B,f)\neq \emptyset$, then $0\in \sharp f(\mathrm{Prune}_{\sharp}(B,f))$;
\item $B \cap \{\vec a\in \mathbb{R}^n: f(\vec a)=0\} \subseteq \mathrm{Prune}_{\sharp}(B, f)$.
\end{enumerate}
\end{definition}
When $\sharp$ is clear, we simply write $\mathrm{Prune}$. The rules can be explained as follows. (W1) ensures that the algorithm always makes progress. (W2) ensures that the result of a pruning is always a reasonable box that may contain a zero, and otherwise $B$ is pruned out. (W3) ensures that the real solutions are never discarded. We proved the following theorem in~\cite{DBLP:conf/cade/GaoAC12}:
\begin{theorem}
Algorithm~\ref{icpalgo} is $\delta$-complete if the pruning operators are well-defined. 
\end{theorem}

\subsection{ODE Pruning in an ICP Framework}

We now study the algorithms for SMT formulas with ODEs. The key is to design the appropriate pruning operators for the solution functions of ODE systems. The pruning operations here strengthen and formalize the ones proposed in~\cite{Eggers2008,Eggers2011,DBLP:conf/cp/GoldsztejnMEH10}, such that $\delta$-completeness can be proved. 

We recall some notations first. Let $D\subseteq \mathbb{R}^n$ be compact and $g_i: D\rightarrow D$ be $n$ Lipschitz-continuous functions. Given the first-order autonomous ODE system
\begin{eqnarray}\label{ivp2}
\frac{d \vec y}{dt} = \vec g(\vec y(t, \vec x_0)) \mbox{ and }\vec y(0, \vec x_0) = \vec x_0
\end{eqnarray}
where $\vec x_0\in D$, we write
\begin{eqnarray*}\label{ivp_solution}
y_i: [0, T]\times D\rightarrow D_i
\end{eqnarray*}
to represent the $i$-th solution function of the ODE system. The $\delta$-regular interval extension of $y_i$ is an interval function
$$\sharp y_i: (\mathbb{IF}\cap [0,T])\times (\mathbb{BF}\cap D) \rightarrow \mathbb{IF}$$
such that for a constant $c\in \mathbb{R}$, for any time domain $I_t\subseteq \mathbb{IF}\cap [0,T]$ and any box of initial values $B_{\vec x_0}\subseteq \mathbb{BF}\cap D$, we have
$$\{x_{t}\in \mathbb{R}: x_{t} = y_i(t, \vec x_0), \vec x_0\in B_{\vec x_0}, t\in I_t\}\subseteq \sharp y_i(I_t, B_{\vec x_0})$$
and
$$|\sharp y_i(I_t, B_{\vec x_0})|\leq \max(c\cdot||I_t\times B_{\vec x_0}||, \delta).$$
We will also need the notion of the {\em reverse} of the ODE system~(\ref{ivp2}), as defined by
\begin{eqnarray}\label{ivpi2}
\frac{d \vec y_-}{dt} = \vec g_-(\vec y_-(t, \vec x_t)) \mbox{ and }\vec y(0, \vec x_t) = \vec x_t.
\end{eqnarray}
Here, $\vec{g}_-$ is defined as $-\vec g$, the vector of functions consisting of the negation of each function in $\vec g$, which is equivalent to reversing time in the flow defined by the ODE system. That is, for $\vec x_0,\vec x_t\in D, t\in \mathbb{R}$, we always have
\begin{eqnarray}
\vec x_t = \vec y (t, \vec x_0)\mbox{ iff }\vec x_0 = \vec y_-(t,\vec x_t).
\end{eqnarray}
Naturally, we write $\sharp (y_-)_i$ to denote the $\delta$-regular interval extension of the $i$-th component of $\vec y_-$.
\begin{algorithm}\label{alg:BasicPruning}
\caption{$\mathrm{ODEPruning}(\sharp{\vec{y}}, \BXz, \BXt, I_t)$}\label{basic}
\begin{algorithmic}[1]
  \Repeat
  \State $\BXt' \gets \prunefwd(\sharp{\vec{y}}, \BXz, \BXt, I_t)$
  \State ${I'_t}   \gets \prunetime(\sharp{\vec{y}}, \BXz, \BXt', I_t)$
  \State $\BXz' \gets \prunebwd(\sharp{\vec{y}}, \BXz, \BXt', I'_t)$
  \Until{$\BXz = \BXz' \land \BXt = \BXt' \land I_t= {I'_t}$}
  \State \Return{$(\BXz', \BXt', {I'_t})$}
\end{algorithmic}
\end{algorithm}

The relation between the initial variables $\vec x_0$, the time duration $t$, and the flow variables $\vec x_t$ is specified by the constraint $\vec x_t = \vec y(t, \vec x_0)$. Given the interval assignment on any two of $\vec x_0$, $\vec x_t$, and $t$, we can use the constraint to obtain a refined interval assignment to the third variable vector. Thus, we can define three pruning operators as follows.
\begin{remark}
The precise definitions of the pruning operators should map the interval assigments on all variables to new assignments on all variables. For notational simplicity, in the pruning operators below we only list the assignments that are actually changed between inputs and outputs. For instance, the forward pruning operator only changes the values on $B_{\vec x_t}$. 
\end{remark}
\noindent{\bf Forward Pruning.} Given interval assignments on $\vec x_0$ and $t$, we compute a refinement of the interval assignments on $\vec x_t$. Figure~\ref{xtp} depicts the forward pruning operation. Formally, we define the following operator:
\begin{definition}[Forward Pruning]
Let $\vec y:[0,T]\times D\rightarrow D$ be the solution functions of an ODE system. Let $B_{\vec x_0}$, $B_{\vec x_t}$, and $I_{t}$ be interval assignments on the variables $\vec x_0$, $\vec x_t$, and $t$. We define the forward-pruning operator as:
$$\mathrm{Prune}_{\mathrm{fwd}}(B_{\vec x_t}, \vec y) = \Hull\Big(B_{\vec x_t}\cap \sharp \vec y(I_t, B_{\vec x_0})\Big).$$
\end{definition}
\begin{figure}
\begin{center}
\includegraphics[width=9cm]{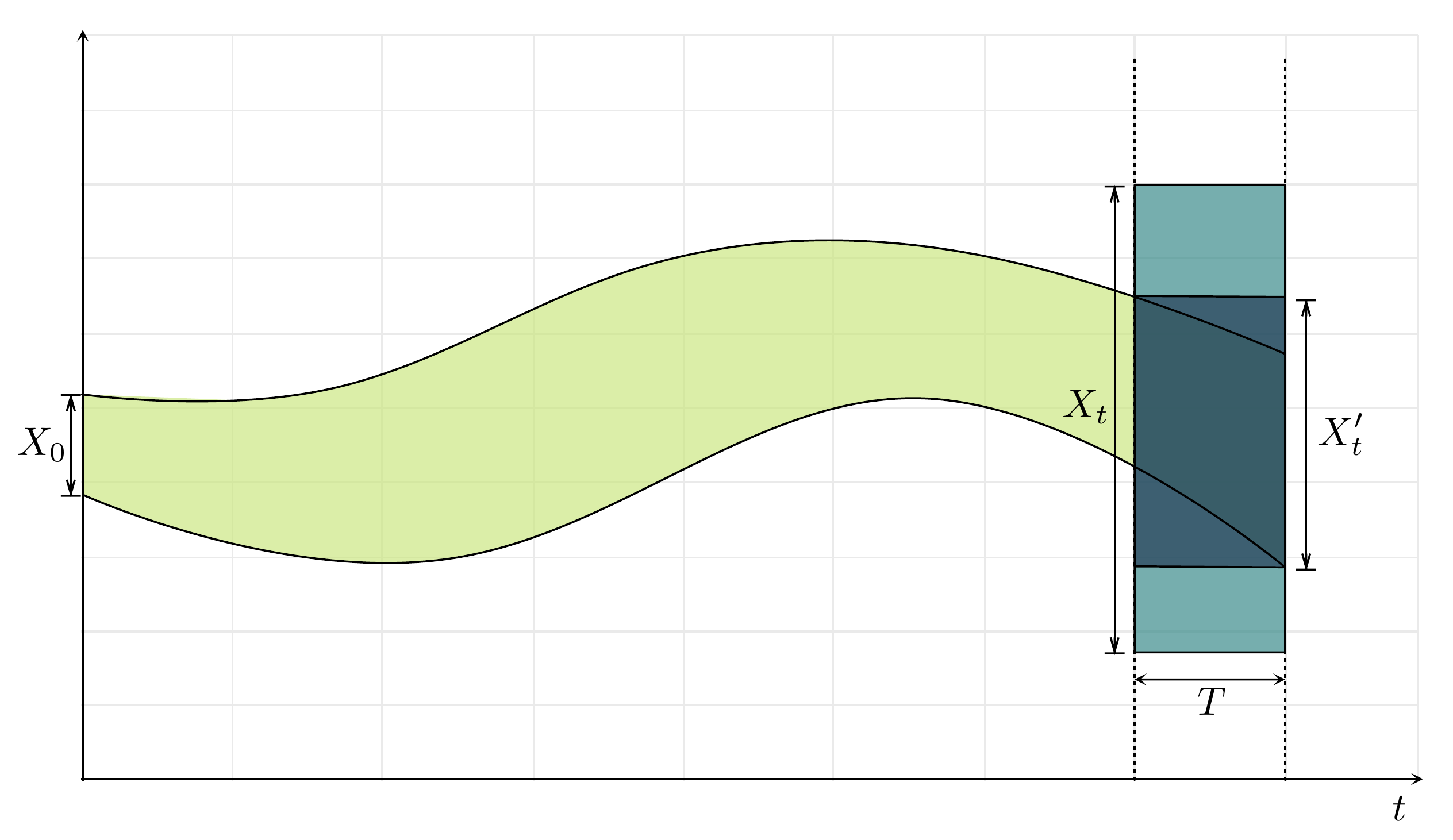}
\end{center}
\caption{Forward Pruning. $X_0$, $X_t$, $t$ represents the current interval assignments, and $X_t'$ is the refined interval assignment on $\vec x_t$ after pruning.}\label{xtp}
\end{figure}

\begin{algorithm}\label{alg:ForwardP}
\caption{$\prunefwd(\sharp{\vec{y}}, \BXz, \BXt, I_t)$}\label{forward}
\begin{algorithmic}[1]
  \State $\BXt' \gets \phi$\;
  \State $I_{{\Delta}t} \gets [I_t^l, I_t^l + {\varepsilon}]$\;
  \While{$I^u_{{\Delta}t} < I_t^u$}
      \State $\BXt' \gets \Hull(\BXt' \cup \sharp{\vec{y}}({{I_{{\Delta}t}}}, \BXz))$
      \State ${{I_{{\Delta}t}}} \gets {I_{{\Delta}t}} + {\varepsilon}$
  \EndWhile
  \State \Return{$\BXt \cap \BXt'$}
\end{algorithmic}
\end{algorithm}

\noindent{\bf Backward Pruning.} Given interval assignments on $\vec x_t$ and $t$, we can compute a refinement of the interval assignments on $\vec x_0$ using the reverse of the solution function. Figure~\ref{x0p} depicts backward pruning. Formally, we define the following operator:
\begin{definition}[Backward Pruning]
Let $\vec y:[0,T]\times D\rightarrow D$ be the solution functions of an ODE system, and let $\vec y_-$ be the reverse of $\vec y$. Let $B_{\vec x_0}$, $B_{\vec x_t}$, and $I_{t}$ be interval assignments on the variables $\vec x_0$, $\vec x_t$, and $t$. We define the backward-pruning operator as:
$$\mathrm{Prune}_{\mathrm{bwd}}(B_{\vec x_0}, \vec y) =\Hull\Big(B_{\vec x_0}\cap \sharp \vec y_-(I_t, B_{\vec x_t})\Big).$$
\end{definition}
\begin{figure}
\begin{center}
\includegraphics[width=9cm]{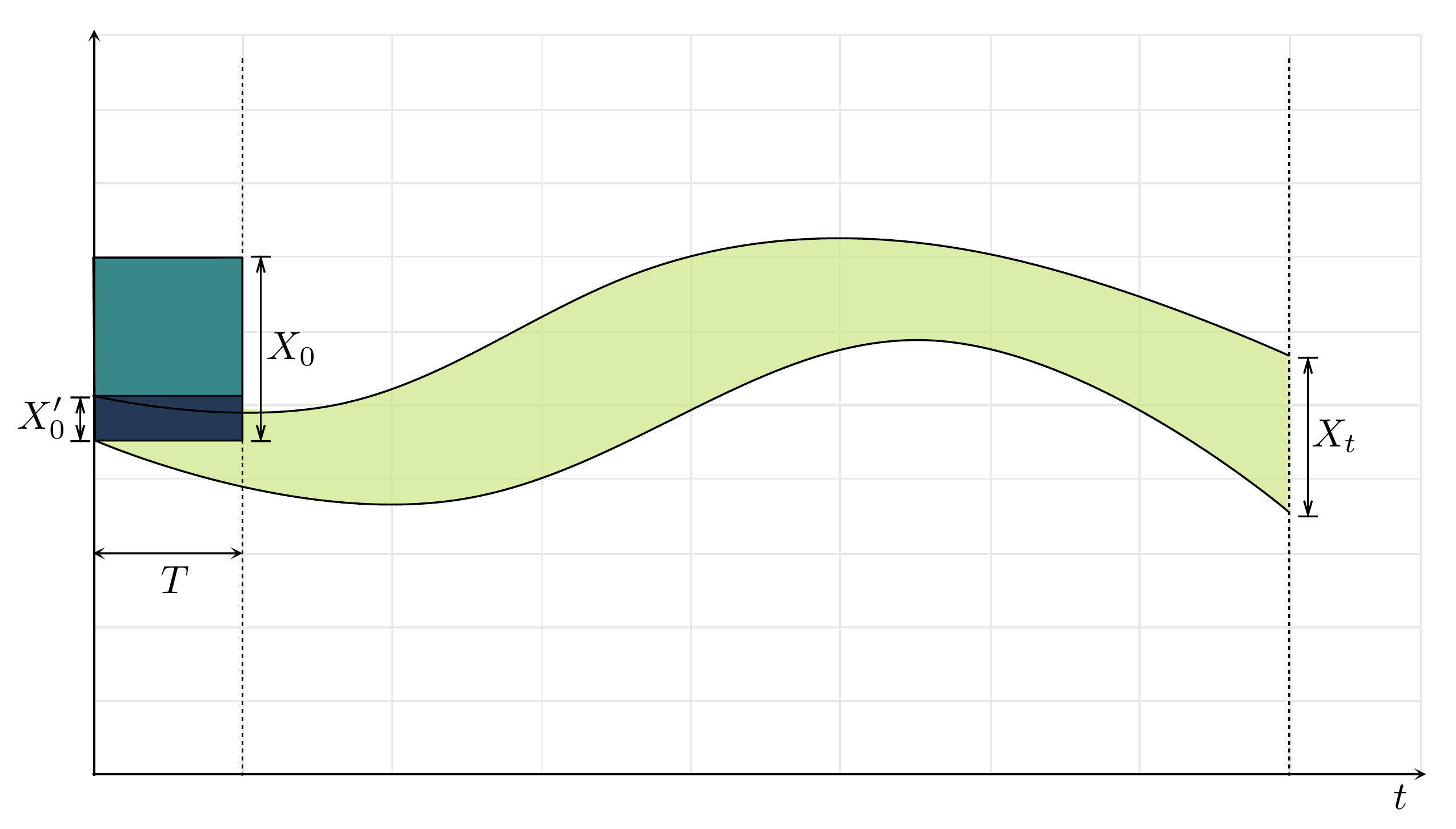}
\end{center}
\caption{Backward Pruning. $X_0$, $X_t$, $t$ represents the current interval assignments, and $X_0'$ is the refined interval assignment on $\vec x_0$ after pruning.
}\label{x0p}
\end{figure}

\noindent{\bf Time-Domain Pruning.} Given interval assignments on $\vec x_0$ and $\vec x_t$, we can also refine the interval assignment on $t$ by pruning out the time intervals that do not contain any $\vec x_t$ that is consistent with the current interval assignments on $\vec x_t$. Figure~\ref{tp} depicts time-domain pruning. Formally, we define the following operator:
\begin{definition}[Time-Domain Pruning]
Let $\vec y:[0,T]\times D\rightarrow D$ be the solution functions of an ODE system. Let $B_{\vec x_0}$, $B_{\vec x_t}$, $I_{t}$ be interval assignments on the variables $\vec x_0$, $\vec x_t$, and $t$. We define the time-domain pruning operator as:
$$\mathrm{Prune}_{\mathrm{time}}(I_{t}, \vec y) =\Hull\Big(I_{t}\cap \{I : \sharp \vec y(I, B_{\vec x_0})\cap  B_{\vec x_t} \not= \emptyset\}\Big).$$

\end{definition}
\begin{figure}
\begin{center}
\includegraphics[width=9cm]{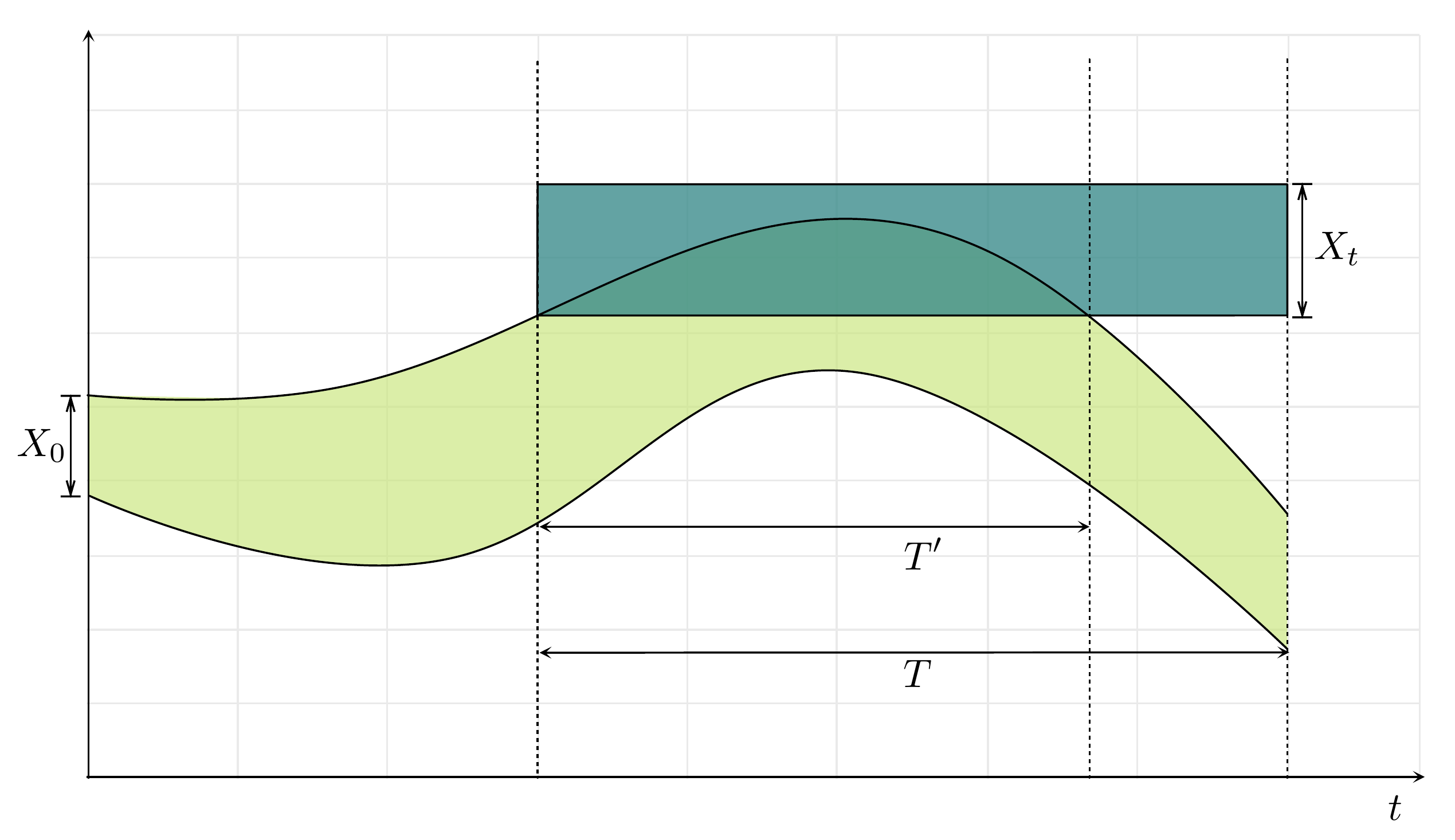}
\end{center}
\caption{Time-Domain Pruning.
 $X_0$, $X_t$, $t$ represents the current interval assignments, and $T'$ is the refined interval assignment on $t$ after pruning.
}\label{tp}
\end{figure}

Overall, the pruning algorithm on based on ODE constraints iteratively applies the three pruning operators until a fixed point on the interval assignments is reached.

\begin{algorithm}\label{alg:BackwardP}
\caption{$\prunebwd(\sharp{\vec{y}}, \BXz, \BXt, I_t)$}\label{backward}
\begin{algorithmic}[1]
  \State $\BXz' \gets \phi$
  \State ${{I_{{\Delta}t}}} \gets [I_t^l, I_t^l + {\varepsilon}]$
  \While{$I^u_{{\Delta}t} < I_t^u$}
      \State $\BXz' \gets \Hull(\BXz' \cup \sharp{\vec{y}_{-}}({{I_{{\Delta}t}}}, \BXt))$
      \State ${{I_{{\Delta}t}}} \gets {I_{{\Delta}t}} + {\varepsilon}$
  \EndWhile
  \State \Return{$\BXz \cap \BXz'$}
\end{algorithmic}
\end{algorithm}

We show the more detailed steps in the three pruning operations in Algorithm~\ref{basic}, \ref{forward}, \ref{backward}, and \ref{timed}.

\begin{algorithm}\label{alg:I_timeP}
\caption{$\prunetime(\sharp{\vec{y}}, \BXz, \BXt, I_t)$}\label{timed}
\begin{algorithmic}[1]
  \State $I_t' \gets \phi$
  \State $I_{{\Delta}t} \gets [I_t^l, I_t^l + {\varepsilon}]$
  \While{$I^u_{{\Delta}t} < I_t^u$}
      \State $\BXt' \gets \sharp{\vec{y}}({{I_{{\Delta}t}}}, \BXz)$
      \If{$\BXt' \cap \BXt \not = \phi$}
          \State $I_t' = \Hull(I_t' \cup {{I_{{\Delta}t}}})$
      \Else
          \State $I_{{\Delta}t} \gets I_{{\Delta}t} + {\varepsilon}$
      \EndIf
  \EndWhile
  \State \Return{$I_t'$}
\end{algorithmic}
\end{algorithm}

\begin{theorem}
The three pruning operators are well-defined.
\end{theorem}

\begin{proof}
We prove that the forward pruning operator is well-defined, and the proofs for the other two operators are similar. Note that the definitions of well-defined pruning are formulated for equality constraints compared to 0. Here we use the function $f = \vec y (t, \vec x_0)-\vec x_t$ in the pruning operator. (Strictly speaking $f$ is a function vector that evaluates to $\vec 0$ on points satisfying the ODE flow. Here for notational simplicity we just write $f$ as a single-valued function and compare with the scalar $0$.)

First, (W1) is satisfied because of the simple fact that for any boxes $B_1,B_2\in \mathbb{BF}$, we have $\Hull(B_1\cap B_2)\subseteq B_1$.

Next, suppose $0\not\in\sharp f(\mathrm{Prune}_{\mathrm{fwd}}(B_{\vec x_t},\vec y)-B_{\vec x_t})$. Then there does not exist any $\vec a_t\in \mathbb{R}^n$ that satisfies both $\vec a_t\in B_{\vec x_t}$ and $\vec a_t\in \mathrm{Prune_{\mathrm{fwd}}}(B_{\vec x_t}, \vec y)$. Since at the same time
$$\mathrm{Prune}_{\mathrm{fwd}}(B_{\vec x_t}, \vec y) = \Hull\Big(B_{\vec x_t}\cap \sharp \vec y(I_t, B_{\vec x_0})\Big)\subseteq B_{\vec x_t},$$
this requires that $\mathrm{Prune}_{\mathrm{fwd}}(B_{\vec x_t}, \vec y) = \emptyset$. Consequently (W2) is satisfied.

Third, note that $\sharp \vec y(I_t, B_{\vec x_0})$ is an interval extension of $\vec y$. Thus, for any $\vec a_t\in \mathbb{R}^n$ such that $\vec y(t, \vec x_0)$ for some $t\in I_t$ and $\vec x_0\in B_{\vec x_0}$, we have $\vec a_t\in \sharp\vec y(I_t, B_{\vec x_0})$. Following the definition of the pruning operator, we have $\vec a_t\in \mathrm{Prune}_{\mathrm{fwd}}(B_{\vec x_t}, \vec y)$. Thus, $B_{\vec x_t}\cap Z_f \subseteq \mathrm{Prune}_{\mathrm{fwd}}(B_{\vec x_t},f)$ and (W3) holds.
\end{proof}

\subsection{$\exists\forall^t$-Formulas and Low-Order Approximations}

For $\exists\forall$-formulas, if the universal quantification is only over the time variables, we can follow the trajectory and prune away the assignment on $\vec x_0$, $\vec x_t$, and $t$ that violates the constraints on the universally quantified time variable. In fact, although the extra quantification complicates the problem, the universal constraints improve the power of the pruning operations.

Here we focus on problems with one ODE system, which can be easily generalized. Let $\vec y$ denote the solution functions of an ODE system, we consider an $\exists\forall^t$-formula of the form
\begin{eqnarray}\label{forall}
\exists^X \vec x_0\exists^X \vec x_t\exists^{[0,T]} t \forall^{[0,t]} t'.\; \vec x_t = \vec y(t, \vec x_0)\wedge \varphi(\vec y (t', \vec x_0))
\end{eqnarray}
Note that the problems encoded as $\Sigma_2$-SMT formulas as listed in Section~\ref{encoding} are all of this form.

We consider $\varphi(\vec y(t', \vec x_0))$ as a special constraint on the $\vec x_0$ and $t$ variables. Using this constraint, we can further refine the three pruning operators as follows.
\begin{definition}[Pruning Refined by $\forall^t$-Constraints]
Let $\vec y:[0,T]\times D\rightarrow \mathbb{R}^n$ be the solution functions of an ODE system. Let $B_{\vec x_0}$, $B_{\vec x_t}$, and $I_{t}$ be interval assignments on the variables $\vec x_0$, $\vec x_t$, and $t$. Let $\varphi(\vec y(t', \vec x_0))$ be a constraint on the universally quantified time variable, as in (\ref{forall}).
We first define
\begin{multline*}\sharp\varphi(I_t, B_{\vec x_0})=\Hull(\{\vec a\in \mathbb{R}^n: \vec a = \vec y(t, \vec x_0), t\in I_t,\\ \vec x_0\in B_{\vec x_0},\mbox{ and }\varphi(\vec a)\mbox{ is true.}\})
\end{multline*}
and define $\sharp\varphi_-$ by replacing $\vec y$ with $\vec y_-$ in the definition above. The forward pruning operator with $\varphi$, written as $\mathrm{Prune}_{\mathrm{fwd}}^{\varphi}(B_{\vec x_t}, \vec y)$, is defined as
$$\Hull\Big(B_{\vec x_t}\cap \sharp \vec y(I_t, B_{\vec x_0})\cap\sharp\varphi(I_t,B_{\vec x_0})\Big)$$
Backward pruning $\mathrm{Prune}^{\varphi}_{\mathrm{bwd}}(B_{\vec x_0}, \vec y)$ is defined as
$$\Hull\Big(B_{\vec x_0}\cap \sharp \vec y_-(I_t, B_{\vec x_t})\cap\sharp\varphi_-(I_t,B_{\vec x_t})\Big).$$
Time-domain pruning $\mathrm{Prune}_{\mathrm{time}}^{\varphi}(I_{t}, \vec y)$ is defined as
$$\Hull\Big(I_{t}\cap \{I : \sharp \vec y(I, B_{\vec x_0})\cap  B_{\vec x_t}\cap\sharp\varphi(I_t,B_{\vec x_0}) \not= \emptyset\}\Big).$$
\end{definition}
In general, $\sharp \varphi$ can be computed by a recursive call to DPLL(ICP), by solving the $\Sigma_1$-SMT problem $\varphi(\vec x)$. In many practical applications, $\varphi$ is of some simple form such as $\vec a\leq \vec x_t\leq \vec b$, in which case simple pruning is shown in Figure~\ref{inv}.
\begin{figure}
\begin{center}
\includegraphics[width=9cm]{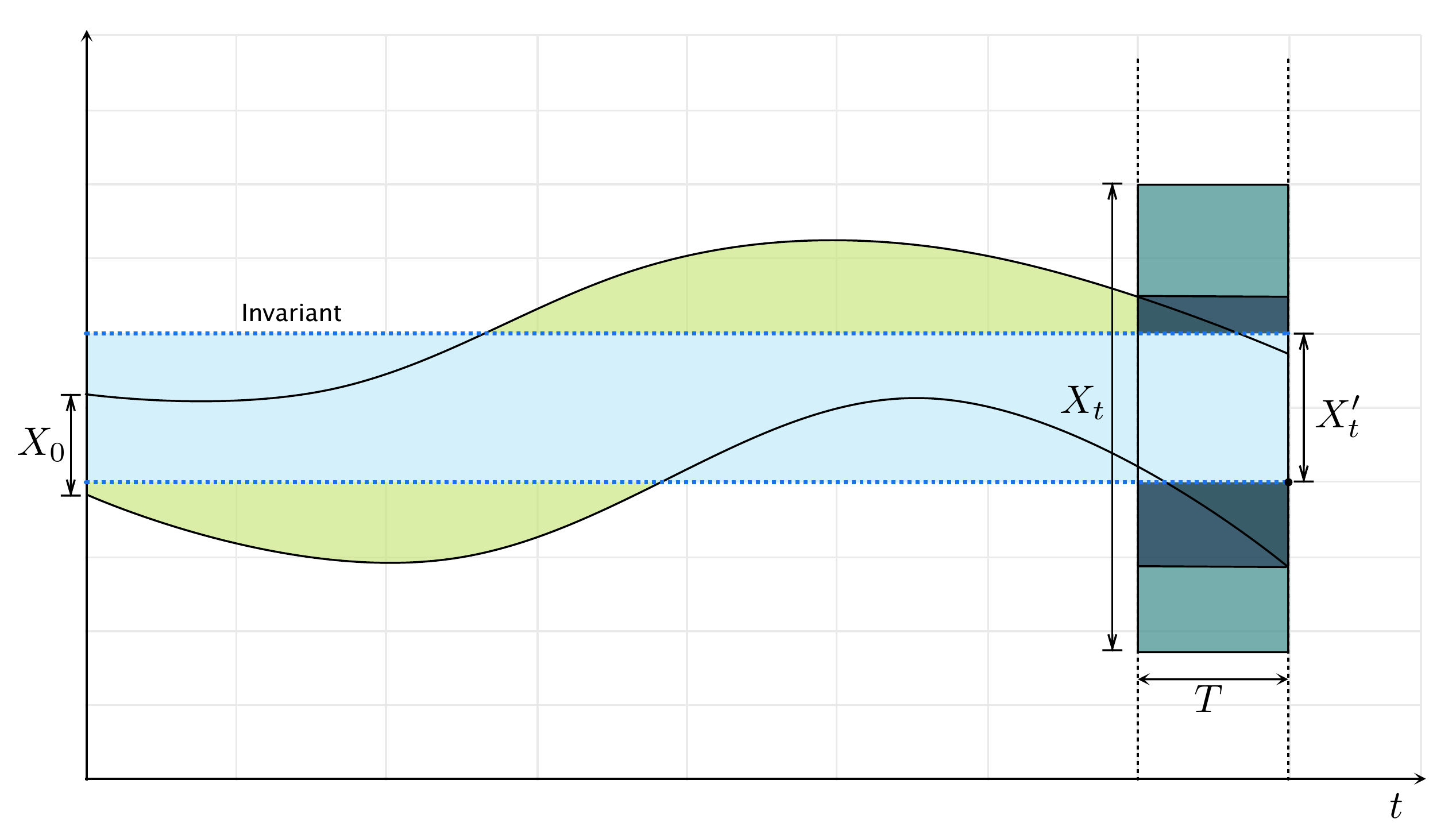}
\end{center}
\caption{Pruning with $\forall^t$-Constraints}\label{inv}
\end{figure}
Another useful heuristic in ODE pruning is to bound the range of the derivatives for a vector space specified by $\vec g$. Suppose for any time $t\in[0,T]$, the derivatives $\vec g$ are bounded in $[\vec l_g, \vec u_g]$. Then by the Picard-Lindel\"of representation, we have
$$\vec x_t = \int_{0}^t \vec g(\vec y(s,\vec y_0))ds + \vec y_0\in [0, T]\cdot [\vec l_g, \vec u_g]+B_{\vec x_0}$$
We can use this formula to perform preliminary pruning on $\vec x_t$, which is especially efficient when combined with $\forall^t$-constraints. Figure~\ref{taylor} illustrates this pruning method.
\begin{figure}
\begin{center}
\includegraphics[width=9cm]{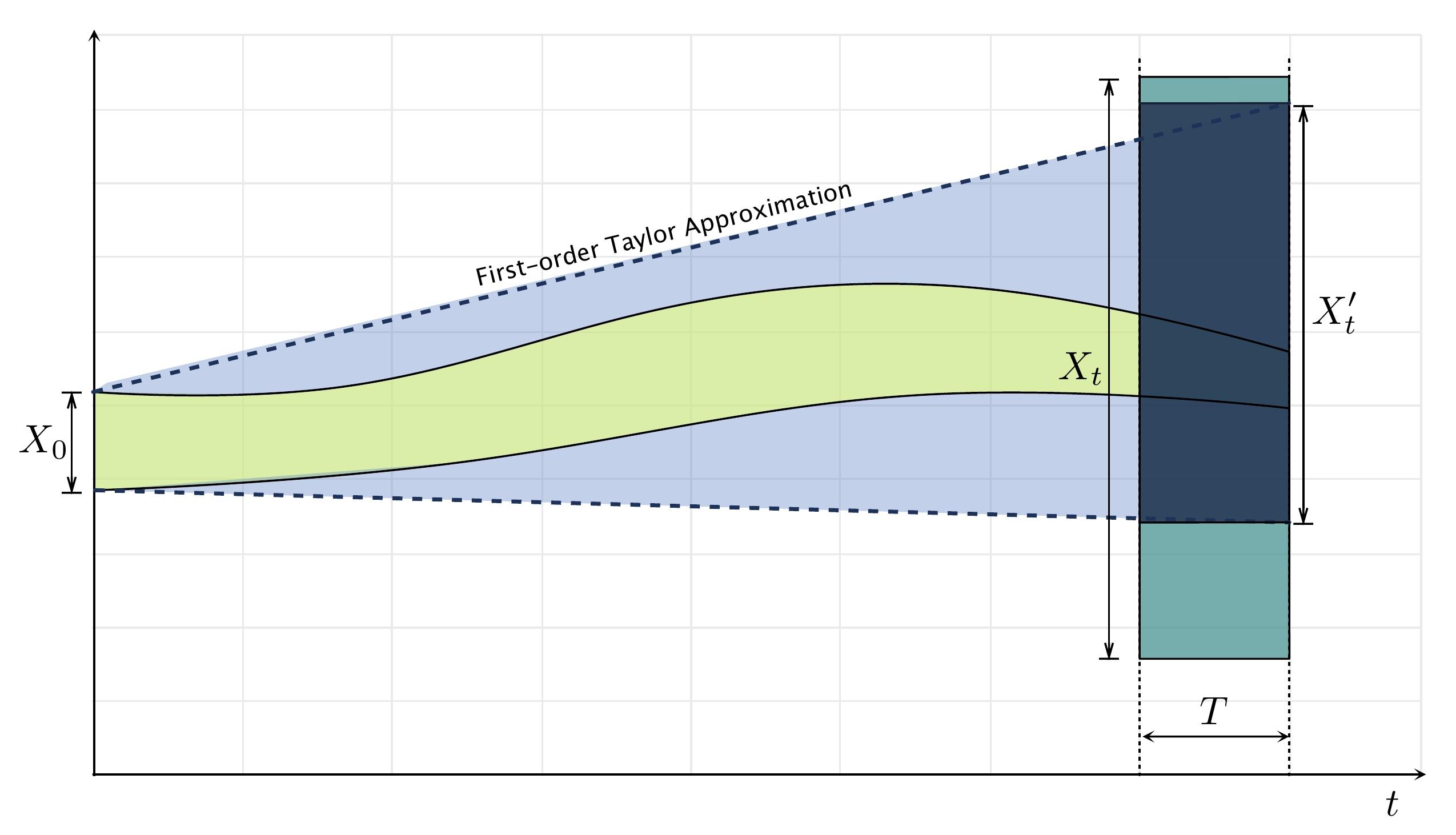}
\end{center}
\caption{Pruning with Low-Order Taylor Approximations}\label{taylor}
\end{figure}

\section{Experiments}\label{experiments}

Our tool {\sf dReal} implements the procedures we studied for solving SMT formulas with ODEs. It is built on several existing packages, including {\sf opensmt}~\cite{DBLP:conf/tacas/BruttomessoPST10} for the general DPLL(T) framework, {\sf realpaver}~\cite{DBLP:journals/toms/GranvilliersB06} for ICP, and {\sf CAPD}~\cite{capd} for computing interval-enclosures of ODEs. The tool is open-source at~\url{http://dreal.cs.cmu.edu}.  All benchmarks and data shown here are also available on the tool website.
\newcommand{\hmodel}[2]{\href{http://dreal.cs.cmu.edu/#1}{#2}}
{\small
\begin{table}[!th]
  \centering
  \small
  \begin{tabular}{l|r|r|r|r|r|r|r|r}
    \hline
    \hline
    P    & \#M& \#D & \#O & \#V  & delta  & R       & Time(s) & Trace \\
    \hline
    \hline
    AF   & 4     & 3        & 20     & 44      & 0.001     & S & 43.10    & 90K      \\
    AF   & 8     & 7        & 40     & 88      & 0.001     & S & 698.86   & 20M      \\
    AF   & 8     & 23       & 120    & 246     & 0.001     & S & 4528.13  & 59M      \\
    AF   & 8     & 31       & 160    & 352     & 0.001     & S & 8485.99  & 78M      \\
    AF   & 8     & 47       & 240    & 528     & 0.001     & S & 15740.41 & 117M     \\
    AF   & 8     & 55       & 280    & 616     & 0.001     & S & 19989.59 & 137M     \\
    \hline
    \hline
    CT     & 2     & 2        & 15     & 36      & 0.005   & S & 345.84   & 3.1M      \\
    CT     & 2     & 2        & 15     & 36      & 0.002   & S & 362.84   & 3.1M      \\
    \hline
    \hline
    EO     & 3     & 2        & 18     & 42      & 0.01    & S & 52.93    & 998K      \\
    EO     & 3     & 2        & 18     & 42      & 0.001   & S & 57.67    & 847K      \\
    EO     & 3     & 11       & 72     & 168     & 0.01    & U & 7.75     & --       \\
    \hline
    \hline
    BB & 2     & 10       & 22     & 66      & 0.01        & S & 0.25     & 123K       \\

    BB & 2     & 20       & 42     & 126     & 0.01        & S & 0.57     & 171K       \\
    BB & 2     & 20       & 42     & 126     & 0.001       & S & 2.21     & 168K       \\
    BB & 2     & 40       & 82     & 246     & 0.01        & U & 0.27     & ----       \\
    BB & 2     & 40       & 82     & 246     & 0.001       & U & 0.26     & ----       \\
    \hline
    \hline
    D1   & 3     & 2        & 9      & 24      & 0.1       & S & 30.84    & 72K      \\
    DU   & 3     & 2        & 6      & 16      & 0.1       & U &  0.04    & --      \\
    \hline
    \hline
  \end{tabular}
  \caption{\small
    \#M = Number of modes in the hybrid system,
    \#D = Unrolling depth,
    \#O = Number of ODEs in the unrolled formula,
    \#V = Number of variables in the unrolled formula,
    R = Bounded Model Checking Result (delta-SAT/UNSAT)
    Time = CPU time (s),
    Trace = Size of the ODE trajectory,
    AF = Atrial Filbrillation,
    CT = Cancer Treatment,
    EO = Electronic Oscillator,
    BB = Bouncing Ball with Drag,
    D1,DU = Decay Model.
}\label{tbl:exp}
\end{table}
}

All experiments were conducted on a machine with a 3.4GHz octa-core Intel Core i7-2600 processor and 16GB RAM, running 64-bit Ubuntu 12.04LTS. Table~\ref{tbl:exp} is a summary of the running time of the tool on various SMT formulas generated from bounded model checking hybrid systems. The formulas typically contain a large number of variables and nonlinear ODEs. 

\begin{figure}[h!]
\begin{center}
\includegraphics[width= 10cm]{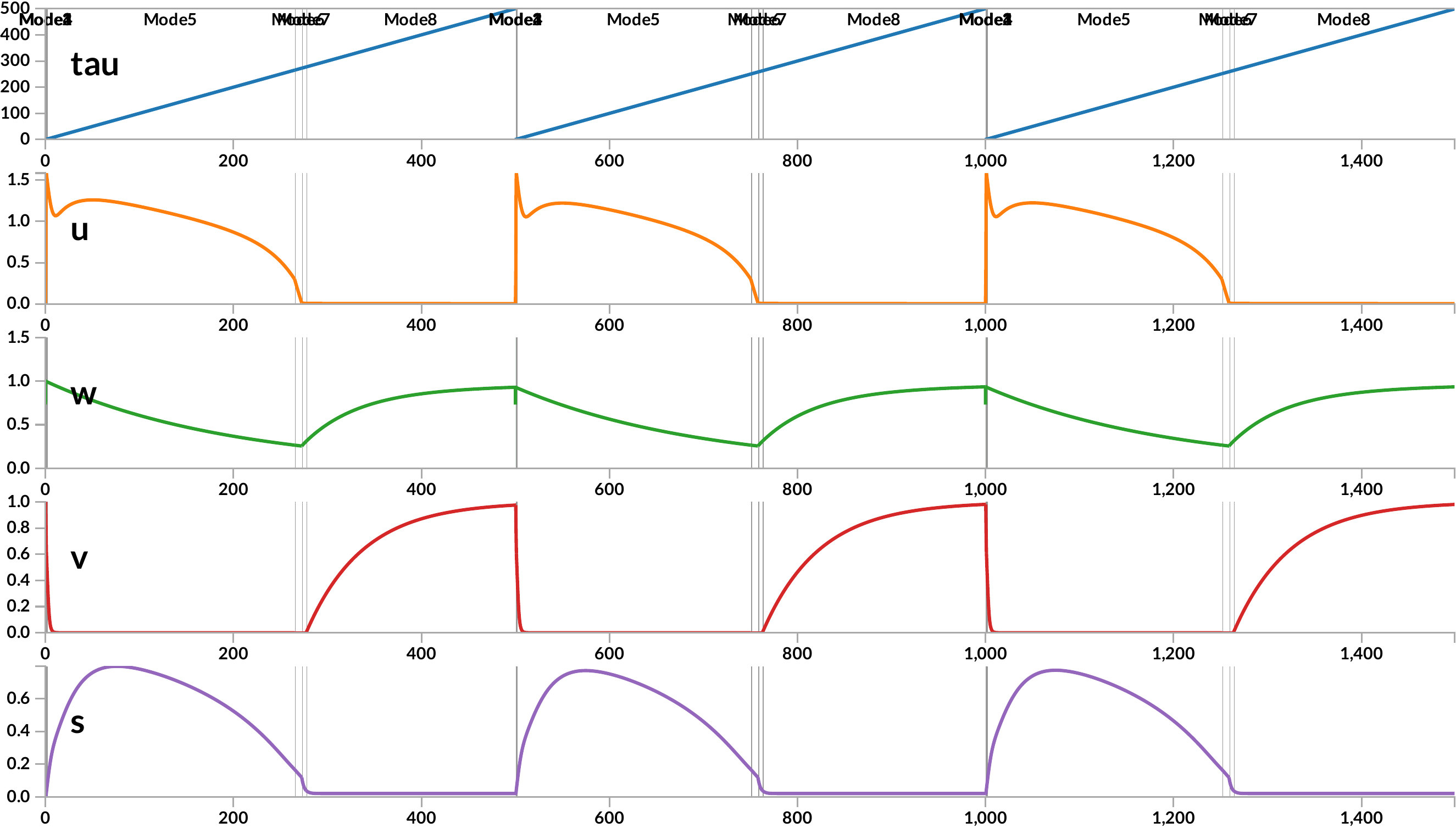}
\includegraphics[width=10cm]{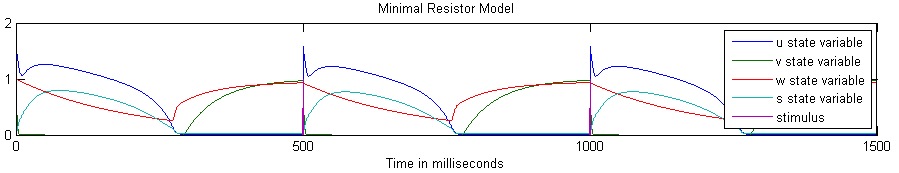}
\end{center}
\caption{Above: Witness for the AF model at depth 23 and 1500 time units. Below: Experimental simulation data.}\label{cardi}
\end{figure}
The AF model as we show in Table~\ref{tbl:exp} is obtained from~\cite{DBLP:conf/cav/GrosuBFGGSB11}. It is a precise model of atrial fibrillation, a serious cardiac disorder. The continuous dynamics in the model concerns four state variables and the ODEs are highly nonlinear, such as:{\small
\begin{eqnarray*}
\frac{du}{dt} &=& e + (u-\theta_v)(u_u-u ) v g_{fi} + wsg_{si}-g_{so}(u)\\
\frac{ds}{dt} &=& \displaystyle\frac{g_{s2}}{(1+e^{-2k(u-us)})} -  g_{s2}s\\
\frac{dv}{dt} &=& -g_v^+\cdot v \ \ \ \ \frac{dw}{dt} = -g_w^+\cdot w
\end{eqnarray*}
}
The exponential term on the right-hand side of the ODE is the sigmoid function, which  often appears in modelling biological switches. On this model, our tool is able to perform a depth-55 unrolling, and solve the generated logic formula. Such a formula contains 280 nonlinear ODEs of the type shown here, with 616 variables. The computed trace from {\sf dReal} suggests a witness of the reachability property that can be confirmed by experimental simulation. Figure~\ref{cardi} shows the comparison between the trace computed from bounded model checking and the actual experimental simulation trace.
\begin{figure}[h!]
\begin{center}
\includegraphics[width=10cm]{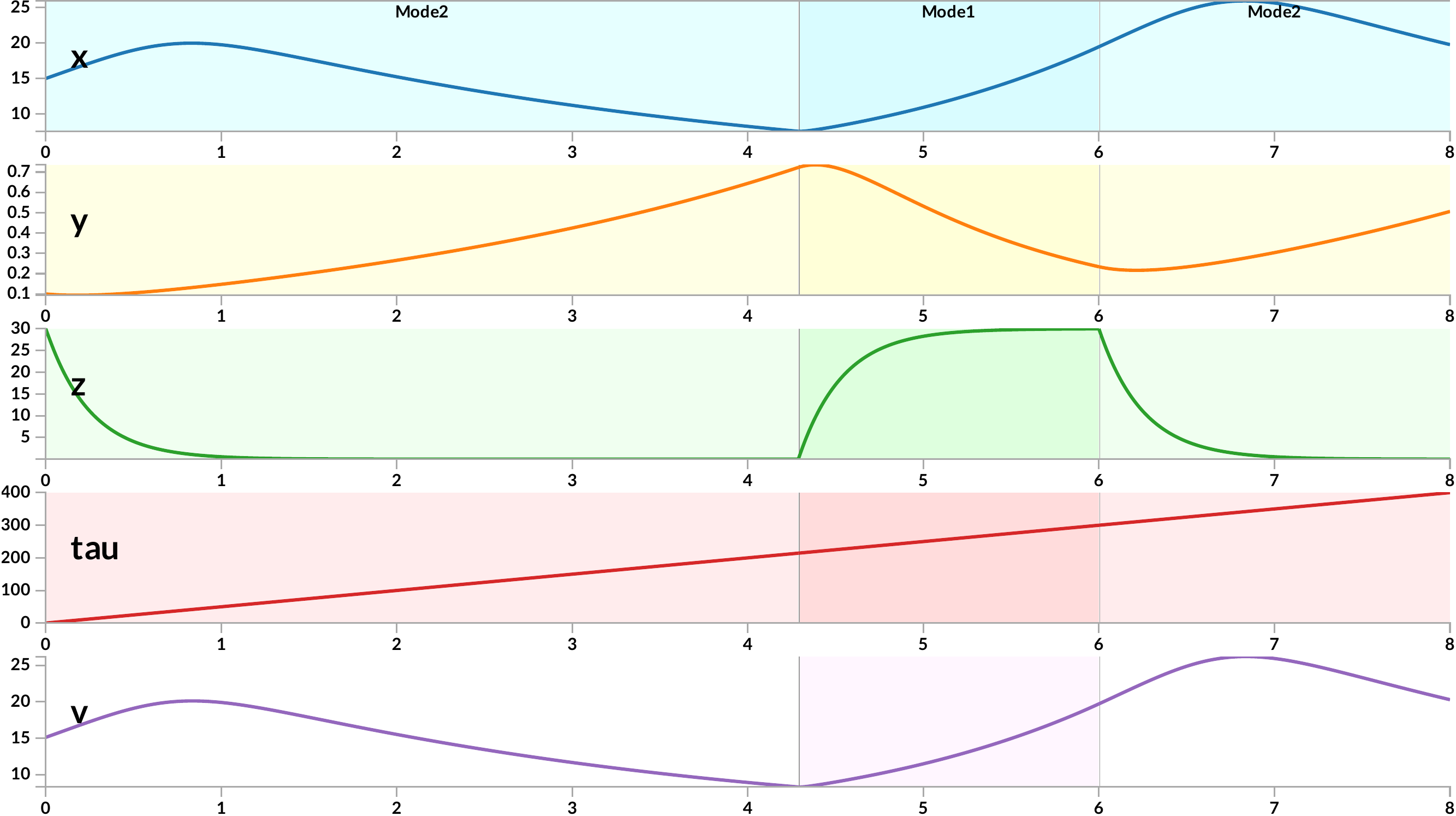}
\includegraphics[width=10cm]{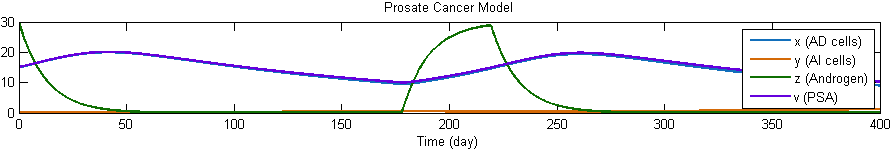}
\end{center}
\caption{Above: Witness computed for the CT model at depth 3 and 500 time units. Below: Experimental simulation data.}\label{ct}
\end{figure}

The CT model represents a prostate cancer treatment model that contains nonlinear ODEs such as following:
{\small
\begin{eqnarray*}
\frac{dx}{dt} &=& (\alpha_x
(k_1+(1-k_1)\frac{z}{z+k_2}\\
& &-\beta_x( (1-k_3)\frac{z}{z+k_4}+k_3)) - m_1(1-\frac{z}{z_0}))x + c_1 x\\
\frac{dy}{dt} &=& m_1(1-\frac{z}{z_0})x+(\alpha_y (1- d\frac{z}{z_0}) - \beta_y)y+c_2y\\
\frac{dz}{dt} &=& \frac{-z}{\tau} + c_3z\\
\frac{dv}{dt} &=& (\alpha_x
(k_1+(1-k_1)\frac{z}{z+k_2}-\beta_x(k_3+(1-k_3)\frac{z}{z+k_4}))\\
& &- m_1(1-\frac{z}{z_0}))x + c_1 x + m_1(1-\frac{z}{z_0})x\\
& &+(\alpha_y (1- d\frac{z}{z_0}) - \beta_y)y+c_2y
\end{eqnarray*}
}The EO model represents an electronic oscillator model that contains nonlinear ODEs such as the following:
{\small
\begin{eqnarray*}
\frac{dx}{dt} &=& - ax \cdot sin(\omega_1 \cdot \tau)\\
\frac{dy}{dt} &=& - ay \cdot sin( (\omega_1 + c_1) \cdot \tau) \cdot sin(\omega_2)\cdot 2\\
\frac{dz}{dt} &=& - az \cdot sin( (\omega_2 + c_2) \cdot \tau) \cdot cos(\omega_1)\cdot 2\\
\frac{\omega_1}{dt} &=& - c_3\cdot \omega_1\ \ \ \frac{\omega_2}{dt} = -c_4\cdot\omega_2\ \ \ \frac{d\tau}{dt} = 1
\end{eqnarray*}
}

The other models are standard simple nonlinear models (for instance, bouncing ball with nonlinear friction), on which our tool has no difficulty in solving.  

\section{Conclusion}\label{conclude}

In this paper we have studied SMT problems over the real numbers with ODE constraints. We have developed $\delta$-complete algorithms in the DPLL(ICP) framework, for both the standard SMT formulas that are purely existentially quantified, as well as $\exists\forall$-formulas whose universal quantification is restricted to the time variables. We have demonstrated the scalability of our approach on nonlinear SMT benchmarks. We believe that the proposed decision procedures can scale on nonlinear problems and can serve as the underlying engine for formal verification of realistic hybrid systems and embedded software.

\noindent{\bf Ackowledgements.} We are grateful for many important suggestions from Jeremy Avigad, Andreas Eggers, and Martin Fr\"anzle. In particular,  we formulated the notion of $\delta$-regular interval extensions to avoid technical difficulties that Eggers and Fr\"anzle pointed out to us. We thank the anonymous referees for various important comments. \bibliographystyle{abbrv}
\bibliography{tau}

\end{document}